\documentclass[lettersize,journal]{IEEEtran}
\usepackage{amsmath,amsfonts,amssymb}
\usepackage{algorithmic}
\usepackage{algorithm}
\usepackage{array}
\usepackage[caption=false,font=normalsize,labelfont=sf,textfont=sf]{subfig}
\usepackage{textcomp}
\usepackage{stfloats}
\usepackage{url}
\usepackage{verbatim}
\usepackage{graphicx}
\usepackage{cite}
\usepackage{orcidlink}
\usepackage{mathtools}
\usepackage{amsthm}
\newtheorem{lemma}{Lemma}
\usepackage{threeparttable}
\usepackage{cases}
\DeclareMathOperator{\Tr}{Tr}
\DeclareMathOperator{\argmin}{argmin}

\begin{document}

\title{Channel Scheduling for IoT Access with\\ Spatial Correlation}

\author{Prasoon Raghuwanshi, ~\IEEEmembership{Student Member,~IEEE},
Onel Luis Alcaraz López\orcidlink{0000-0003-1838-5183},~\IEEEmembership{Member,~IEEE},
Petar Popovski\orcidlink{0000-0001-6195-4797},~\IEEEmembership{Fellow,~IEEE},
Matti Latva-aho\orcidlink{0000-0002-6261-0969},~\IEEEmembership{Senior Member,~IEEE}
\thanks{Prasoon Raghuwanshi, Onel Luis Alcaraz López, and Matti Latva-aho are with the Centre for Wireless Communications, University of Oulu, $90570$, Oulu, Finland (e-mail: Prasoon.Raghuwanshi@oulu.fi; Onel.AlcarazLopez@oulu.fi; Matti.Latva-aho@oulu.fi).}
\thanks{Petar Popovski is with the Department of Electronic Systems,	Aalborg University, $9220$, Aalborg, Denmark (e-mail: petarp@es.aau.dk)}
\thanks{This research has been supported by the Research Council of Finland (former Academy of Finland) 6G Flagship Programme (Grant 346208), the Finnish Foundation for Technology Promotion, and the INDIFICORE project.}} 



\maketitle
\vspace{-4mm}
\begin{abstract}
Spatially correlated device activation is a typical feature of the Internet of Things (IoT). This motivates the development of channel scheduling (CS) methods that mitigate device collisions efficiently in such scenarios, which constitutes the scope of this work. Specifically, we present a quadratic program (QP) formulation for the CS problem considering the joint activation probabilities among devices. This formulation allows the devices to stochastically select the transmit channels, thus, leading to a soft-clustering approach. We prove that the optimal QP solution can only be attained when it is transformed into a hard-clustering problem, leading to a pure integer QP, which we transform into a pure integer linear program (PILP). We leverage the branch-and-cut (B$\&$C) algorithm to solve PILP optimally. Due to the high computational cost of B$\&$C, we resort to some sub-optimal clustering methods with low computational costs to tackle the clustering problem in CS. Our findings demonstrate that the CS strategy, sourced from B$\&$C, significantly outperforms those derived from sub-optimal clustering methods, even amidst increased device correlation.
\end{abstract}
\vspace{-4mm}
\section{Introduction}
\IEEEPARstart{T}{he} Internet of Things (IoT) synergizes connectivity, automation, and data-driven insights to enable businesses to make informed decisions, reduce waste, and enhance productivity. Indeed, IoT networks will embrace a multitude of services, fostering seamless interaction among numerous devices like sensors and actuators \cite{10097471}. This makes low latency, massive connectivity support, and energy efficiency very desirable features for an IoT multiple-access scheduling protocol.

IoT devices are often deployed for tasks where data is transmitted only upon detecting a specific event, making communication in such IoT networks event-triggered. Such networks typically experience sporadic and bursty traffic involving short data packets, while requiring efficient random access (RA) mechanisms at the devices to contend for available channels in a distributed manner. Notably, a defining attribute of event-triggered networks is spatially correlated device activation. This phenomenon arises because devices close to an event epicenter are more likely to simultaneously detect and respond to the same event. Such correlation must be exploited by proper channel scheduling (CS) procedures for RA \cite{9458251}.

The research community has recently displayed a growing interest in utilizing (deep) reinforcement learning (DRL/RL) based CS algorithms for IoT with correlated device activation, e.g., in \cite{rech2021coordinated, chandak2022learning}.
Specifically, the authors in \cite{rech2021coordinated} approached the time slot selection task at each industrial IoT device as a Markov game and employed the linear reward-inaction RL algorithm \cite{4082268} to determine its equilibrium points. Nevertheless, such a proposal does not converge to a pure Nash equilibrium and is effective only under moderate traffic intensity.
Meanwhile, the authors in \cite{chandak2022learning} proposed a multi-armed bandit (MAB) agent at each device to determine its CS strategy autonomously. Over time, MABs train themselves distributively and generate a form of implicit coordination among CS strategies of the devices that contribute to successfully transmitting an alarm message. However, due to instability, only a single MAB can be trained at a time. This constraint significantly prolongs the training time and necessitates tight synchronization among devices.
Moreover, DRL/RL algorithms impose a significant computational burden on resource-constrained IoT devices. Consequently, CS at the base station (BS), leveraging spatially correlated device activation, holds greater promise for IoT networks. Indeed, after the IoT devices sense an event, they no longer need to run a channel scheduling algorithm but just exploit the scheduling decisions already made by the BS. Consequently, the overall latency of the RA procedure is reduced, as well as the energy consumption at the IoT device, thereby extending its battery life.

\begin{figure}[!t]
\centering
\includegraphics[width=0.9\linewidth]{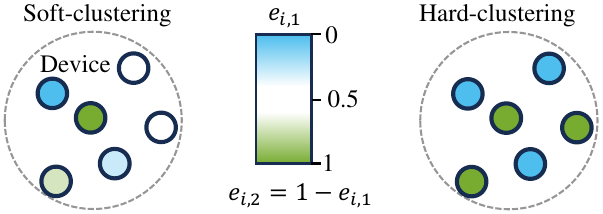}
\vspace{-2mm}
\caption{Illustration of the soft-/hard-clustering-based CS exemplified for the case of $6$ devices and $2$ channels. Here, $e_{i,1}$ and $e_{i,2}$ represent the probabilities of device $i$ choosing channel-$1$ and channel-$2$, respectively. As can be seen, devices with a high likelihood of simultaneous activation are opting for different transmission channels.}
\label{clustering_types}
\vspace{-4mm}
\end{figure}

The implementation of CS at the BS was studied in \cite{8445866, 9458251, lopez2022statistical}, where the authors leverage the correlated device activation to design a time slot/channel allocation protocol. 
The non-convex nature of the time slot allocation problem, aimed at maximizing throughput, coupled with the incorporation of higher-order correlation, prompted the authors in \cite{8445866} to derive the expressions for the lower/upper bounds on throughput using the inclusion-exclusion principle \cite{flajolet_sedgewick_2009}. Based on these bounds, fresh non-convex problems that tap into the pairwise correlation among devices are formulated. These problems are then solved using heuristic algorithms to derive a sub-optimal CS strategy with near-optimal throughput. 
Meanwhile, the time slot allocation problem in \cite{9458251} converges to a locally optimal solution when solved using the stochastic gradient ascent algorithm, whose performance is sensitive to the initialization of the time slot allocation matrix. Furthermore, the CS strategy in \cite{9458251} requires information about the device activity in each frame. 
On the other hand, a hard-clustering CS strategy is proposed in \cite{lopez2022statistical}, but no optimality guarantees are provided.

Motivated by these inadequacies, we propose a novel CS method customized for IoT networks with spatially correlated device activations. Specifically, we first introduce a quadratic program (QP) formulation for the CS problem while exploiting the joint activation probabilities among devices. This ensures that devices with a high likelihood of simultaneous activation opt for different channels for transmission, as shown in Fig.~\ref{clustering_types}. Although the CS approach does not restrict the devices to using deterministically the same channels, we prove that such static scheduling is indeed optimal. Specifically, we prove that the optimal solution for QP is attained after the latter is transformed into a hard-clustering problem, i.e., a pure integer QP (PIQP). Fig.~\ref{clustering_types} illustrates the difference between the soft-/hard-clustering in a simple setup. To simplify PIQP further, we transform it into an equivalent pure integer linear program (PILP). 
We employ branch-and-cut (B$\&$C) \cite{basu2021complexity} algorithm to optimally solve PILP. Notably, due to the high computational cost of B$\&$C, we also resort to sub-optimal clustering methods with low computational cost to tackle the CS problem, including K-Medoids \cite{bhat2014k} (as in \cite{lopez2022statistical}) and its variant K-Medoids with K-Means++ \cite{10.5555/1283383.1283494} initialization, and discuss the performance tradeoffs of the different solvers. 
Finally, the performance of our proposal is analyzed numerically, and our findings demonstrate that the CS strategy derived from B$\&$C outperform significantly the one obtained from the sub-optimal clustering methods even when there is a surge in the device correlation. Also, B$\&$C achieves a near-optimal objective value of PILP within a time duration that closely matches the solving duration of sub-optimal clustering approaches, thus evincing the suitability of the former.

The paper is structured as follows. Section ~\ref{SysModProb} presents the system model and introduces the formulation for the CS problem. Section ~\ref{OptiCS} proves that hard-clustering-based scheduling is indeed optimal and presents an optimal CS strategy. Section ~\ref{nume_results} presents the numerical results and some sub-optimal clustering methods to tackle the CS problem. Finally, in Section ~\ref{Concl}, we conclude and highlight some open research areas.

In the remaining of the paper, $\Tr(\cdot)$ signifies the trace of a square matrix, $[\cdot]^T$ denotes the transpose operation, and boldface lowercase/uppercase letters are used to indicate column vectors/matrices.

\vspace{-3mm}
\section{System Model and Problem Formulation} \label{SysModProb}
\vspace{-1mm}
Consider a set ${\mathcal{N}=\{1,2,\dots,N\}}$ of $N$ sporadically active IoT devices that exhibit correlated activation patterns. These devices communicate with the BS using a set ${\mathcal{L}=\{1,2,\dots,L\}}$ of $L$ orthogonal channels, where ${L < N}$. Additionally, assume that the BS knows the pairwise joint activation probability matrix ${\mathbf{A} \in [0,1]^{N \times N}}$, where each element $A_{i_1,i_2}$ represents the joint activation probability of devices ${i_1, i_2 \in \mathcal{N}}$. Notice that $\mathbf{A}$ is symmetric. 

Let, $e_{i,j}$ represent the probability of device ${i \in \mathcal{N}}$ choosing channel ${j \in \mathcal{L}}$, thus ${\sum_{j=1}^L e_{i,j}=1}$. Also, let us denote by ${\mathcal{N}_j \subset \mathcal{N}}$ the set of devices that selected a common channel ${j \in \mathcal{L}}$ in a given time, thus ${\cup_{j \in \mathcal{L}} \mathcal{N}_j = \mathcal{N}}$ and ${\mathcal{N}_{j_1} \cap \mathcal{N}_{j_2} = \varnothing, \ \forall j_1 , j_2 \in \mathcal{L}, \ j_1 \neq j_2}$. In this setup, active devices transmit in a grant-free manner, and collisions between two transmitting devices occur only when they both utilize the same channel. Consequently, we can define the collision probability for the channel ${j \in \mathcal{L}}$ as
\begin{align}
    P_j = 1 - \!\!\!\!\!\!\!\! \prod_{i_1, i_2 \in \mathcal{N}_j , i_1 < i_2} \!\!\!\!\!\!\!\! (1 - A_{i_1,i_2}e_{i_1,j}e_{i_2,j}) ,
\end{align}
and the overall network's average collision probability as
\begin{equation}
    \begin{aligned} \label{expanded_colli_prob}
        P_c & = \frac{1}{L} \sum_{j \in \mathcal{L}} P_j ,\\
        & = \frac{1}{L} \sum_{j \in \mathcal{L}} \Bigl[ \{ \cdot \}_j + \!\! \sum_{i_1 \in \mathcal{N}_j} \sum_{i_2 \in \mathcal{N}_j, i_1 < i_2} \!\!\!\!\!\!\! A_{i_1,i_2}e_{i_1,j}e_{i_2,j} \Bigr] ,
    \end{aligned}
    \vspace{-1mm}
\end{equation}
where $\{ \cdot \}_j, \forall j \in\mathcal{L}$ consists of higher-order terms. Observe that $(\ref{expanded_colli_prob})$ depends entirely on $\mathbf{A}$, $L$, and ${e_{i,j}, \forall i \in\mathcal{N}, \forall j \in\mathcal{L}}$.

The CS objective is to find matrix ${\mathbf{E} = [\mathbf{e}_{1}, \cdots, \mathbf{e}_{N}]^T \in [0,1]^{N \times L}}$, where ${\mathbf{e}_i = [e_{i,1}, \cdots, e_{i,L}]^T}$, that minimizes $P_c$. Such information is later forwarded to the devices by the BS. The following is a formulation for the CS problem
\begin{subequations}\label{opti_problem_Pc}
\begin{align}
\label{obj_Pc} \underset{\mathbf{E}}{\textrm{minimize}} & \ P_c \\
\label{1st__Pc} \textrm{subject to} & \ \sum_{j=1}^{L} e_{i,j} = 1  \textrm{,} \ \forall i \in\mathcal{N} , \\
\label{2nd__Pc} & \ 0 \leq e_{i,j} \leq 1, \forall i \in\mathcal{N}, \forall j \in\mathcal{L} .
\end{align}
\end{subequations}
Constraints $(\ref{1st__Pc})$ and $(\ref{2nd__Pc})$ enforce that the optimization variables are indeed probability values, as per their definition. Furthermore, problem $(\ref{opti_problem_Pc})$ enables devices to select channels probabilistically, making it akin to a soft-clustering problem. However, observe that problem $(\ref{opti_problem_Pc})$ is a non-convex optimization problem, since the objective function $(\ref{obj_Pc})$ is non-convex. This makes problem $(\ref{opti_problem_Pc})$ difficult to solve. The following section will address the hitches in problem $(\ref{opti_problem_Pc})$ without any working assumption on the construction properties of $\mathbf{A}$.
\vspace{-2mm}
\section{Problem Reformulation and Proposed CS} \label{OptiCS}
First, we will transform problem $(\ref{opti_problem_Pc})$ into a QP using the following result proved in Lemma \ref{LeMma1}.
\begin{lemma} \label{LeMma1}
The upper bound on $P_c$ can be expressed as
\vspace{-1mm}
\begin{align} \label{Pc_upper_bound}
    P_c \leq \frac{1}{L} \sum_{j \in \mathcal{L}} \sum_{i_1 \in \mathcal{N}_j} \sum_{i_2 \in \mathcal{N}_j, i_1 < i_2} \!\!\!\!\!\!\! A_{i_1,i_2}e_{i_1,j}e_{i_2,j} = F(\mathbf{E}).
\end{align}
Also, $F(\mathbf{E})$ tightly approaches $P_c$ as $\{ \cdot \}_j \rightarrow 0, \forall j \in\mathcal{L}$ in $(\ref{expanded_colli_prob})$.
\end{lemma}
\begin{proof}
We begin by expressing $P_j$ as the probability of the union of collision events where each collision event $C_{i_1,i_2}^{(j)}$ corresponds to a collision between device $i_1$ and $i_2$, ${i_1 < i_2}$, ${i_1, i_2 \in \mathcal{N}_j}$, in channel $j$
\begin{align}
    P_j = P \Bigl(\!\!\!\! \bigcup_{i_1, i_2 \in \mathcal{N}_j} \!\!\!\! C_{i_1,i_2}^{(j)}\Bigr) .
\end{align}
Meanwhile, Boole's inequality \cite{SENETA199224}, commonly known as union bound, states that for any finite set of events ${C_{i_1,i_2}^{(j)}, \forall i_1, i_2 \in \mathcal{N}_j}$, the probability of their union is bounded by the sum of their individual probabilities
\begin{align}
    P \Bigl(\!\!\!\! \bigcup_{i_1, i_2 \in \mathcal{N}_j} \!\!\!\! C_{i_1,i_2}^{(j)} \Bigr) \leq \sum_{i_1 \in \mathcal{N}_j} \sum_{i_2 \in \mathcal{N}_j, i_1 < i_2} \!\!\!\!\!\!\! P(C_{i_1,i_2}^{(j)}) .
\end{align}
Since ${P(C_{i_1,i_2}^{(j)}) = A_{i_1,i_2}e_{i_1,j}e_{i_2,j}}$, one obtains
\begin{align} \label{Pj_upper_bound}
    P_j \leq \sum_{i_1 \in \mathcal{N}_j} \sum_{i_2 \in \mathcal{N}_j, i_1 < i_2} \!\!\!\!\!\!\! A_{i_1,i_2}e_{i_1,j}e_{i_2,j} .
\end{align}
Therefore, substituting $(\ref{Pj_upper_bound})$ into $(\ref{expanded_colli_prob})$ yields $(\ref{Pc_upper_bound})$. Furthermore, ${F(\mathbf{E})}$ can serve as a tight upper bound as ${\{ \cdot \}_j \rightarrow 0, \forall j \in\mathcal{L}}$, in $(\ref{expanded_colli_prob})$, a condition achievable only as $N/L$ diminishes. A decrease in $N/L$ leads to a decrease in ${\lvert \mathcal{N}_j \rvert, \forall j \in\mathcal{L}}$, consequently causing a decline in the number of terms within ${\{ \cdot \}_j, \forall j \in\mathcal{L}}$, subsequently diminishing its magnitude.
\end{proof}
This leads to the following CS problem formulation
\begin{subequations}\label{orig_opti}
\begin{align}
\label{obj} \underset{\mathbf{E}}{\textrm{minimize}} & \ F(\mathbf{E}) \\
\label{1st_constraint} \textrm{subject to} & \ \sum_{j=1}^{L} e_{i,j} = 1  \textrm{,} \ \forall i \in\mathcal{N} , \\
\label{2nd_constraint} & \ 0 \leq e_{i,j} \leq 1, \forall i \in\mathcal{N}, \forall j \in\mathcal{L} .
\end{align}
\end{subequations}
Observe that minimizing the objective function $(\ref{obj})$ enforces that devices with a high joint activation probability choose different channels, while those with a low joint activation probability are prone to choose the same channel. 

Let us remodel $(\ref{obj})$ as $F(\mathbf{E}) = \frac{1}{2L} \Tr(\mathbf{E}^{T} \mathbf{A} \mathbf{E})$.
We know that $\Tr(\mathbf{E}^{T} \mathbf{A} \mathbf{E})$ is convex only if $\mathbf{A}$ is positive semi-definite, which may not be the case. Therefore, although optimization tools like successive convex approximation may solve it, they do not guarantee an optimal solution. The performance of successive convex approximation depends on the quality of the initial feasible point. 
Observe that problem $(\ref{orig_opti})$ is a QP comprising $NL$ optimization variables.

The devices may not necessarily utilize the same channels consistently over time. However, the following result proves that the optimal solution for QP can only be attained when it is transformed into a hard-clustering problem, which signifies that such static scheduling is indeed optimal.
\begin{lemma} \label{LeMma2}
The optimal solution of problem $(\ref{orig_opti})$ is discrete, i.e., ${e_{i,j} \in \{0, 1\}, \forall i \in\mathcal{N}, \forall j \in\mathcal{L}}$.
\end{lemma}
\begin{proof}
Employing the \textit{proof by contradiction} method, we begin by positing that a soft-clustering strategy, denoted as ${e_{i,j} \in [0, 1), \forall i \in\mathcal{N}, \forall j \in\mathcal{L}}$, minimizes $F(\mathbf{E})$. With this premise in mind, we proceed by extracting terms from $F(\mathbf{E})$ that pertain to a generic device $\hat{i}$,
\vspace{-2mm}
\begin{align*}
    F(\mathbf{E}) = & \frac{1}{L} \Bigl[ \sum_{\substack{i_1=1 \\ i_1 \neq \hat{i}}}^{N} \sum_{\substack{i_2=1 \\ i_2 \neq \hat{i} \\ i_1 < i_2}}^{N} \sum_{j=1}^{L} A_{i_1,i_2} e_{i_1,j} e_{i_2,j} + F'(\mathbf{E}) \Bigr] ,
\end{align*}
\vspace{-3mm}
where
\begin{align*}
    F'(\mathbf{E}) & = \sum_{\substack{i_2=1 \\ \hat{i} < i_2}}^{N} \sum_{j=1}^{L} A_{\hat{i},i_2} e_{\hat{i},j} e_{i_2,j} 
    + \sum_{\substack{i_1=1 \\ i_1 < \hat{i}}}^{N} \sum_{j=1}^{L} A_{i_1,\hat{i}} e_{i_1,j} e_{\hat{i},j} ,\\
    & = \sum_{j=1}^{L} e_{\hat{i},j} \!\!\!\!\! \overbrace{\sum_{i_2=1, \hat{i} < i_2}^{N} \!\!\!\!\!\! A_{\hat{i},i_2} e_{i_2,j}}^{w_1(j)} 
    + \sum_{j=1}^{L} e_{\hat{i},j} \!\!\!\!\! \overbrace{\sum_{i_1=1, i_1 < \hat{i}}^{N} \!\!\!\!\!\! A_{i_1,\hat{i}} e_{i_1,j}}^{w_2(j)} \\
    & = \sum_{j=1}^{L} e_{\hat{i},j} (w_1(j) + w_2(j)) .
    \vspace{-2mm}
\end{align*}
Given constraint $(\ref{1st_constraint})$, $F'(\mathbf{E})$ is minimized when ${e_{\hat{i},j} = 1}$ for $j = {\argmin_{\j} \{w_1(\j)+w_2(\j)\}}$ and ${e_{\hat{i},j} =0}$ otherwise.
This reasoning is valid for any device ${i\in\mathcal{N}}$, since $\hat{i}$ is a generic device.
Hence, our presumption that the soft-clustering strategy could steer us towards the minimum value of $F(\mathbf{E})$ is indeed flawed. As a result, we can conclude that the minimum value of $F(\mathbf{E})$ can only be attained when ${e_{i,j} \in \{0, 1\}, \forall i \in\mathcal{N}, \forall j \in\mathcal{L}}$.
\vspace{-3mm}
\end{proof}
\subsection{Integer Program-based Reformulation}
\vspace{-1mm}
Introducing binary constraints on $e_{i,j}$ converts problem $(\ref{orig_opti})$ into a PIQP. Specifically, constraint $(\ref{2nd_constraint})$ can be substituted by ${e_{i,j} \in \{0, 1\}, \forall i \in\mathcal{N}, \forall j \in\mathcal{L}}$. This means that $e_{i,j}$ denotes the assignment of device $i$ to cluster $j$, while constraint $(\ref{1st_constraint})$ ensures that each device is assigned to exactly one cluster.

Algorithms like B$\&$C and branch-and-bound (B$\&$B) can solve integer programs. Specifically, both algorithms calculate an upper bound on the optimal objective function value by solving the relaxed problem of the corresponding integer program. Importantly, if this relaxed problem happens to be convex, it can yield an ideal upper bound. Conversely, attempting to solve a non-convex relaxation using heuristic methods may result in a mediocre upper bound. 

Recall that the QP relaxation of PIQP can only be convex when matrix $\mathbf{A}$ is positive semidefinite. Therefore, we next transform PIQP into an equivalent PILP, which offers two key advantages: i) the relaxation of a PILP is a linear, thus convex, program (LP), and ii) the convexity of the LP relaxation is independent of the construction properties of $\mathbf{A}$. Bearing this in mind, let us rewrite $F(\mathbf{E})$ with ${\mathbf{E} \in \{0,1\}^{N \times L}}$ as
\begin{align} \label{remaining_objective_equation}
    F(\mathbf{E}) = \frac{1}{L} \sum_{i_1=1}^{N} \sum_{\substack{i_2=1 \\ i_1 < i_2}}^{N} \sum_{j=1}^{L} A_{i_1,i_2} z_{i_1,i_2,j} ,
\end{align}
where ${z_{i_1,i_2,j} = \max \{e_{i_1,j} + e_{i_2,j}-1, 0\}}$, which can be substituted by the following set of linear constraints
\begin{align} \label{remaining_equations}
    \begin{array}{c}
        \forall i_1,i_2 \in\mathcal{N}, \\ 
        i_1 < i_2, \\ 
        \forall j \in\mathcal{L} 
        \end{array} & \begin{cases} 
            e_{i_1,j} + e_{i_2,j} - z_{i_1,i_2,j} \leq 1 , \\
            z_{i_1,i_2,j} \geq 0 , \\
            z_{i_1,i_2,j} + y_{i_1,i_2,j} \leq e_{i_1,j} + e_{i_2,j} , \\
            z_{i_1,i_2,j} - y_{i_1,i_2,j} \leq 0 , \\
            y_{i_1,i_2,j} \in \{ 0, 1 \}.
        \end{cases}
\end{align}
Therefore, the PIQP can be transformed into the following PILP involving ${LN^2}$ optimization variables
\begin{align} \label{transformed_opti}
    \begin{split}
        \underset{\mathbf{E}, \ z_{i_1,i_2,j}, \ y_{i_1,i_2,j}
        }{\textrm{minimize}} & (\ref{remaining_objective_equation}) \\
        \textrm{subject to} \ \quad & (\ref{1st_constraint}) , (\ref{remaining_equations}), \\
        & e_{i,j} \in \{0, 1\},  \forall i \in\mathcal{N},  \forall j \in\mathcal{L} ,
    \end{split}
\end{align}
which can be solved using B$\&$B or B$\&$C algorithms.
%
\vspace{-3mm}
\subsection{PILP Solvers and Insights}
We define the complexity of B$\&$B and B$\&$C algorithms in terms of the number of offspring relaxed subproblems they solve to achieve the optimal solution. The complexity of B$\&$B scales exponentially with the number of optimization variables increases, which is motivated by the fact that the upper and lower bounds of the optimal objective value, determined during the B$\&$B run, are not sufficiently tight to discard a substantial proportion of offspring subproblems. Meanwhile, the complexity of B$\&$C is inherently lower due to its ability to generate exceptionally tight bounds for the optimal objective value by employing the cutting plane method during its B$\&$B procedure. This tightness facilitates the elimination of a considerable number of offspring subproblems. Hence, when confronted with problems involving a large number of optimization variables, B$\&$C proves to be less time-consuming than B$\&$B. Still, B$\&$C may be computationally intense, thus motivating the exploration of sub-optimal but low-complexity clustering methods in the next section.

{
\setlength\arrayrulewidth{1pt}
\begin{table}[!t]
\caption{Sub-optimal Clustering Methods \label{SOptiClus}}
\vspace{-2mm}
\centering
\begin{tabular}{p{1.29cm}@{\hskip 0.1in} p{6.7cm}}
\hline
\textbf{Method} & \textbf{Summary} \\
\hline
K-Medoids & Uses medoid, a cluster member, as cluster representative, where Initial medoids are selected uniformly at random. Then, an iterative process begins where non-medoid devices are assigned to the existing medoids, and a new medoid is determined. Finally, the process ends when medoids no longer change. \\
K-Medoids with K-Means++ & K-Means++ initialize the medoids for the K-Medoids by randomly selecting the first medoid and subsequent medoids based on probability. The K-Medoids then refine the medoids and form clusters. \\
\hline
\end{tabular}
\vspace{-6mm}
\end{table}
}

\vspace{-2mm}
\section{Performance Evaluation} \label{nume_results}
Consider a circular region where the devices are uniformly distributed with a fixed density of $0.2$ devices$/$m$^2$. As a result, we adjust the radius of this area to accommodate the desired number of devices. We leverage Gurobi solver \cite{nonconvex_quadratic_program} to implement B$\&$C. Moreover, B$\&$C terminates its run when the gap between the lower and upper bound of the optimal objective value, determined during its run, is less than $1 \%$.

We compute $\mathbf{A}$ using the alarm generation approach, where we continuously generate alarms over $10^6$ time steps and collect data on active devices. After its completion, we compute $\mathbf{A}$, where $A_{i_1,i_2}$, ${\forall i_1, i_2 \in \mathcal{N}}$, ${i_1<i_2}$, represents the probability of device $i_1$ and $i_2$ became active simultaneously in the recently concluded alarm generation procedure. Notably, an alarm is triggered at a randomly chosen location, known as the epicenter, at each time step. This alarm activates a random set of devices based on the activation probability function ${f(d_i) = e^{-d_i/\lambda}}$, where $d_i$ is the distance of device $i$ from the epicenter, and $\lambda = 3$ serves as a mean-scaling multiplier. 
\vspace{-2mm}
\subsection{Sub-optimal Clustering Methods}
Numerous approaches have been put forward in the extensive literature on clustering methods, each possessing its strengths and limitations. For the device clustering scenario considered in problem $(\ref{transformed_opti})$, we exploit the K-Medoids method and its variant K-Medoids with K-Means++ initialization. These methods are briefly described in Table \ref{SOptiClus}.
\begin{figure}[!t]
\centering
\vspace{-2mm}
\includegraphics[width=3in]{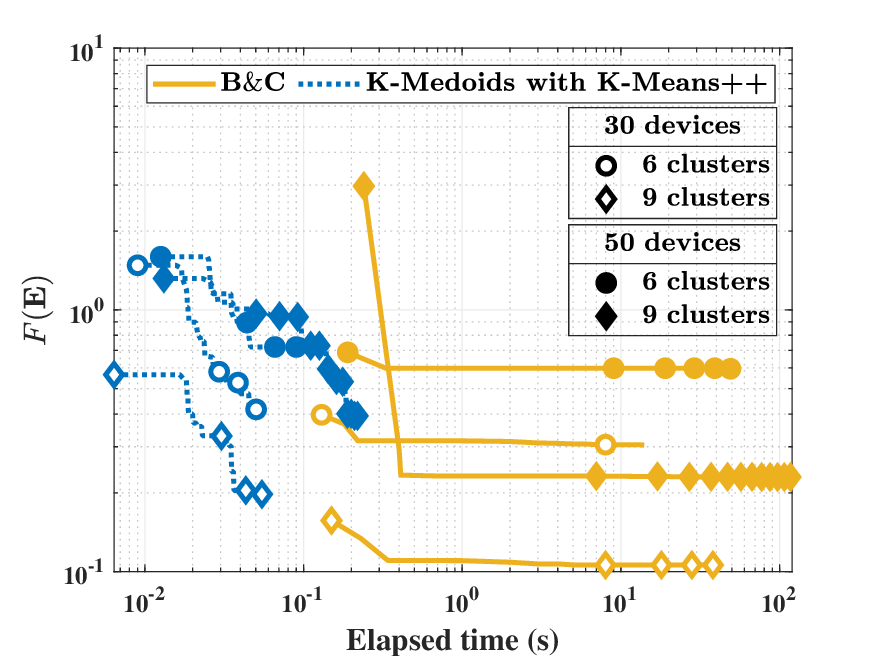}
\vspace{-2mm}
\caption{Logs from various algorithms showcasing the best-known objective value over time for problem $(\ref{transformed_opti})$.}
\label{objective_vs_time}
\vspace{-3mm}
\end{figure}
\vspace{-3mm}
\subsection{Numerical Results}
\vspace{-1mm}
Fig. \ref{objective_vs_time} showcases the log data from the B$\&$C and K-Medoids with K-Means++ algorithms, depicting the best-known objective value's progression over time for the problem $(\ref{transformed_opti})$. Notably, B$\&$C achieves a near-optimal objective value within a time duration that closely matches the solving duration of K-Medoids with K-means++ for all combinations of $N$ and $L$. Moreover, B$\&$C is computationally more intense than K-Medoids with K-means++, as numerous offspring relaxed subproblems have to be solved in B$\&$C. Consequently, the solving duration of B$\&$C would be higher than K-Medoids with K-means++, as visible in Fig. \ref{objective_vs_time}. However, the solution obtained from B$\&$C is superior.

Fig. \ref{Pc_vs_F(E)} validates $F(\mathbf{E})$ as an upper bound of $P_c$. The visual representation in Fig. \ref{Pc_vs_F(E)} clearly demonstrates that for all values of $L$, the relationship $P_c \leq F(\mathbf{E})$ holds. Additionally, as $N/L$ decreases, $F(\mathbf{E})$ becomes an increasingly accurate approximation of $P_c$. This enhanced precision arises from the fact that ${\{ \cdot \}_j, \forall j \in\mathcal{L}}$, in $(\ref{expanded_colli_prob})$ decreases with a decrement in $N/L$. Meanwhile, Fig. \ref{Pc_vs_F(E)} also illustrates that the increment in $L$ makes it feasible to achieve remarkably low $P_c$ through the solution obtained by solving the problem $(\ref{transformed_opti})$ using B$\&$C.

Fig. \ref{collisionProbability_vs_devices} illustrates a comparison in terms of $P_c$ between the discussed solvers for problem $(\ref{transformed_opti})$. Regardless of $N$ and $L$, $P_c$ obtained from B$\&$C consistently outperforms the ones obtained from K-Medoids with K-Means++. Additionally, a crucial observation from Fig. \ref{collisionProbability_vs_devices} is that increasing the value of $L$ decreases $P_c$ obtained from both B$\&$C and K-Medoids with K-Means++. However, K-Medoids with K-Means++ exhibits a diminishing marginal benefit from increasing $L$, wherein the positive impact on $P_c$ diminishes with each successive increment of $L$. Consequently, $P_c$ from K-Medoids with K-Means++ never reaches $10^{-2}$. 

Fig. \ref{correlation_vs_N} illustrates the performance of the discussed solvers across various levels of correlation among device activation. According to $f(d_i)$, device activation probability increases with an increase in $\lambda$, which in turn boosts the aforesaid correlation. As anticipated, Fig. \ref{correlation_vs_N} distinctly illustrates an increase in $\lambda$ leads to a noticeable rise in $P_c$ for both solvers. Furthermore, it is noteworthy that, irrespective of the value of $\lambda$, $P_c$ obtained from B$\&$C consistently outperforms the ones obtained from K-Medoids with K-Means++. However, a significant observation from Fig. \ref{correlation_vs_N} is the diminishing relative advantage of B$\&$C over K-Medoids with K-Means++ as $\lambda$ continues to increase.

\begin{figure}[!t]
\centering
\vspace{-2mm}
\includegraphics[width=3in]{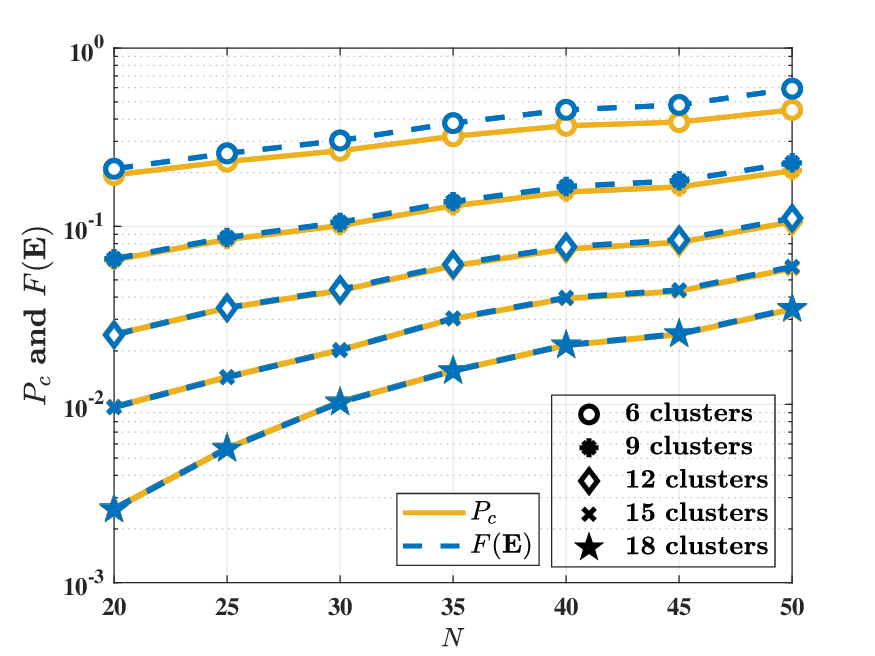}
\vspace{-2mm}
\caption{$P_c$ and $F(\mathbf{E})$, obtained by solving problem $(\ref{transformed_opti})$ using B$\&$C, as a function of $N$.}
\label{Pc_vs_F(E)}
\vspace{-3mm}
\end{figure}

\begin{figure}[!t]
\centering
\vspace{-2mm}
\includegraphics[width=3in]{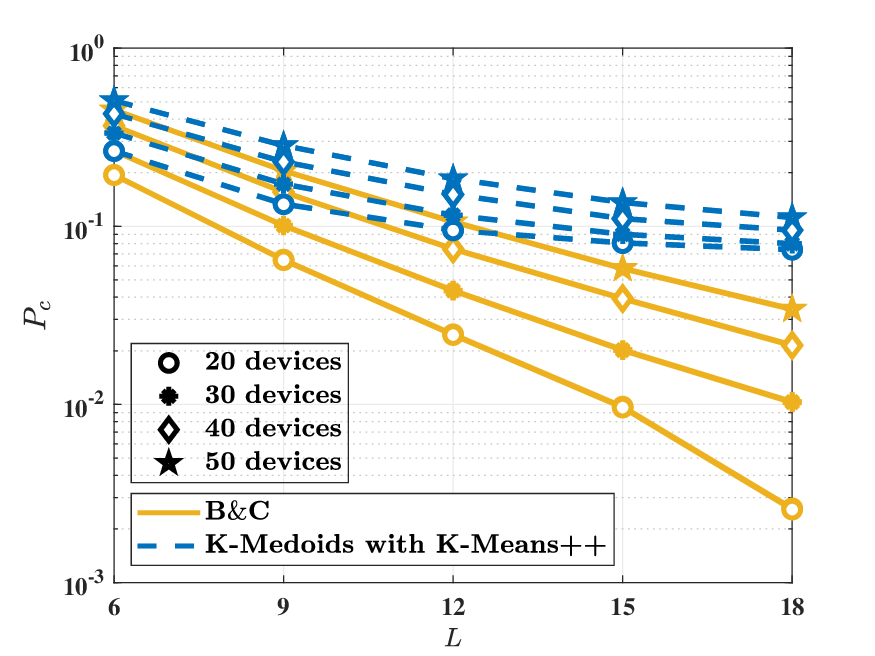}
\vspace{-2mm}
\caption{$P_c$, obtained by solving problem $(\ref{transformed_opti})$ using B$\&$C and K-Medoids with K-Means++, as a function of $L$.}
\label{collisionProbability_vs_devices}
\vspace{-5mm}
\end{figure}

\section{Conclusion} \label{Concl}
\vspace{-1mm}
This letter presents a novel CS method for IoT applications involving spatially correlated device activation. First, we introduce a QP formulation for the CS problem, which considers joint activation probabilities among devices. We then establish the compelling insight that solving QP optimally necessitates its transformation into a hard-clustering problem called PIQP. This PIQP is subsequently transformed into a PILP. 
To optimally solve PILP, we leverage the B$\&$C algorithm. Given the high computational cost of B$\&$C, we also discuss some sub-optimal clustering methods to tackle the CS problem. 

The problem $(\ref{transformed_opti})$ addressed in this study necessitates prior specification of the number of clusters. Hence, an intriguing avenue for future research is to develop a clustering optimization problem that eliminates the need to specify the number of clusters beforehand, followed by its subsequent solution.
Another compelling research avenue lies in performing intra-cluster scheduling following the device clustering based on $\mathbf{A}$. Intra-cluster scheduling ensures that not all simultaneously active devices within the same cluster need to transmit. Notably, the cluster head oversees the intra-cluster scheduling.

\begin{figure}[!t]
\centering
\vspace{-2mm}
\includegraphics[width=3in]{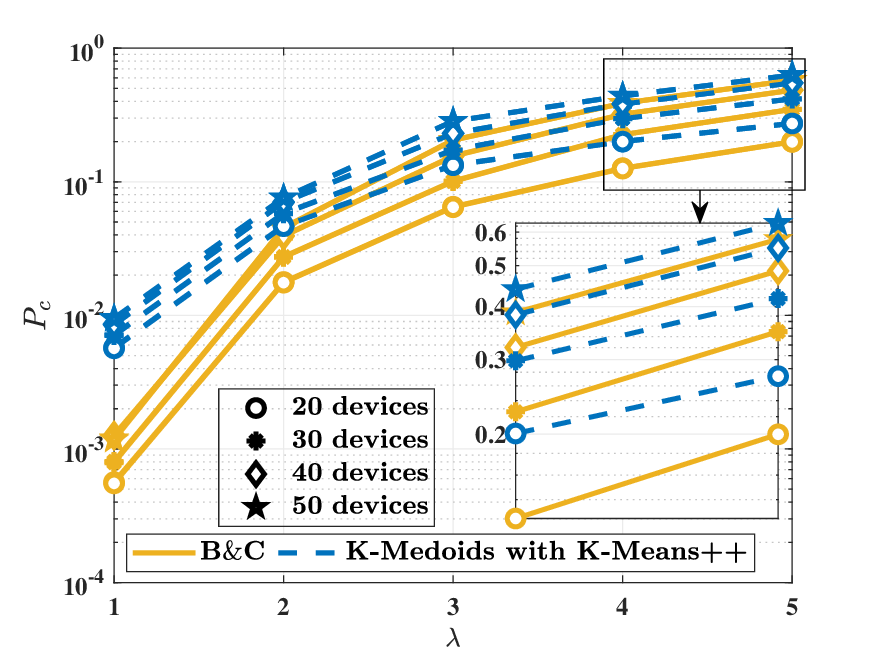}
\vspace{-2mm}
\caption{$P_c$, obtained by solving problem $(\ref{transformed_opti})$ using B$\&$C and K-Medoids with K-Means++, as a function of $\lambda$ for $L = 9$.}
\label{correlation_vs_N}
\vspace{-5mm}
\end{figure}


\bibliographystyle{IEEEtran}
\bibliography{IEEEabrv,references}

\newpage






\vfill

\end{document}